\documentclass[runningheads]{llncs}

\usepackage{amsbsy}
\usepackage{amssymb}
\usepackage{amsmath}
\usepackage{mathtools}
\usepackage{booktabs}
\usepackage{graphicx}
\usepackage{hyperref}
\usepackage{wrapfig}
\usepackage{enumitem}
\usepackage{multirow}
\usepackage{scrextend}
\usepackage{algorithm}
\usepackage[english]{babel}
\usepackage[utf8]{inputenc}
\usepackage[noend]{algpseudocode}

\usepackage{tikz}
\usetikzlibrary{patterns,quotes,angles}

\begin{document}
\title{On Singleton Congestion Games with\\Resilience Against Collusion\thanks{This work is supported by The Scientific and Technological
Research Council of Turkey (TÜBİTAK) through grant 118E126.}}

\titlerunning{On Singleton Congestion Games with Resilience Against Collusion}


\author{
Bugra Caskurlu\inst{1} \and
\"{O}zg\"{u}n Ekici\inst{2}
\and Fatih Erdem Kizilkaya\inst{1}}

\authorrunning{B. Caskurlu \and \"{O}. Ekici \and F. E. Kizilkaya}

\institute{
TOBB University of Economics and Technology, Ankara, Turkey\\
\email{bcaskurlu@etu.edu.tr},
\email{f.kizilkaya@etu.edu.tr}\\
\and
\"{O}zye\u{g}in University, Istanbul, Turkey\\
\email{ozgun.ekici@ozyegin.edu.tr}}

\maketitle

\begin{abstract}
We study the subclass of singleton congestion games with identical and increasing cost functions, i.e., each agent tries to utilize from the least crowded resource in her accessible subset of resources. Our main contribution is a novel approach for proving the existence of equilibrium outcomes that are resilient to weakly improving deviations: $(i)$ by singletons (Nash equilibria), $(ii)$ by the grand coalition (Pareto efficiency), and $(iii)$ by coalitions with respect to an \textit{a priori} given partition coalition structure (partition equilibria). To the best of our knowledge, this is the strongest existence guarantee in the literature of congestion games that is resilient to weakly improving deviations by coalitions.
\end{abstract}

\section{Introduction}
Game forms are useful mathematical abstractions to analyze behavior in real-life multi-agent systems. The use of game forms in analyzing multi-agent systems can be traced back to the seminal work of Wardrop \cite{Wardrop}, who modeled traffic flow in transportation networks using a game form. We follow this line of research by studying a game form which is motivated by navigational multi-agent systems. We namely study a subclass of singleton congestion games \cite{SCG}, where all cost functions of resources are identical (and increasing), which we refer to as \textit{identical singleton congestion games} (ISCGs). An ISCG is a simple game form in which there are $n$ agents and $m$ resources such that each agent tries to utilize from the least crowded resource in her accessible subset of resources.

This game form captures the spirit of agent interactions in real-life scenarios in which there is collision or interference, as is typical in many domains. For instance, in the field of multi-agent navigation, collision avoidance is a fundamental issue, due to the need for autonomous control of potentially a large number of robots running in the same workspace; see \cite{CollisionAvoidance}. Figure \ref{figure:CollisionAvoidance} below presents a motivating toy example, in which robots navigate in a $2$-dimensional workspace by moving to a neighboring cell at each iteration. Under the natural assumption that collision is more likely in a cell with a higher number of robots, robots may be thought of as playing an ISCG at each single iteration.

\begin{figure}[h]
\begin{center}
\scalebox{0.66}{
\begin{tikzpicture}
\draw[very thick] (0, 0) -- (8, 0);
\draw[very thick] (0, 8) -- (8, 8);
\draw[very thick] (0, 0) -- (0, 8);
\draw[very thick] (8, 0) -- (8, 8);

\draw[dotted] (0, 2) -- (8, 2);
\draw[dotted] (0, 4) -- (8, 4);
\draw[dotted] (0, 6) -- (8, 6);
\draw[dotted] (2, 0) -- (2, 8);
\draw[dotted] (4, 0) -- (4, 8);
\draw[dotted] (6, 0) -- (6, 8);

\draw[very thick] (4, 8) -- (4, 4);
\draw[very thick] (4, 4) -- (6, 4);
\draw[very thick] (0, 2) -- (2, 2);
\draw[very thick] (2, 2) -- (2, 6);
\draw[very thick] (6, 0) -- (6, 2);

\fill[black] (1, 7.4) circle (2pt);
\fill[black] (0.5, 6.6) circle (2pt);
\fill[black] (1.5, 6.6) circle (2pt);

\fill[black] (4.6, 0.6) circle (2pt);
\fill[black] (4.6, 1.4) circle (2pt);
\fill[black] (5.4, 0.6) circle (2pt);
\fill[black] (5.4, 1.4) circle (2pt);

\fill[black] (4.6, 5.4) circle (2pt);
\fill[black] (5.4, 4.6) circle (2pt);

\fill[black] (7, 7) circle (2pt);

\fill[black] (0.6, 1.4) circle (2pt);
\fill[black] (1.4, 0.6) circle (2pt);
\end{tikzpicture}}
\end{center}
\caption{A simple example of a multi-agent navigation problem where dots represent the robots and bold lines represent obstacles.}
\label{figure:CollisionAvoidance}
\end{figure}
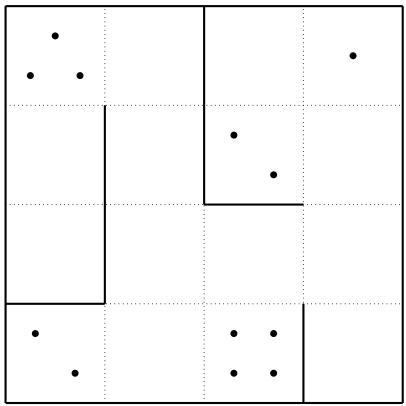

Though ISCGs is a more restricted subclass of congestion games compared to the recent game forms studied in the literature (for example, see \cite{ProjectGames}); there is merit in studying this game form, both due to its applications in the field of multi-agent navigation, and existence guarantee of ``very stable" strategy profiles. To the best of our knowledge, the existence guarantee that we present in this paper is the strongest one in the congestion games literature resilient against weakly improving deviations. It is also known that there does not exist such an existence guarantee in singleton congestion with ``almost'' identical (and increasing) cost functions, even if each resource is accessible by all agents \cite{Arxiv}.


\bigskip

On the technical side, we incorporate a novel approach for proving the aforementioned existence guarantee. As typical in many existence guarantee proofs, we essentially show that any maximal element of an asymmetric and transitive relation defined over the outcomes of the game is always ``stable'', which ensures the existence guarantee. This is almost always done by means of a potential function argument, i.e., by arguing that if an agent (or a coalition in our case) has a ``deviation'' in an outcome, then the resulting outcome must be greater with respect to an asymmetric and transitive relation \cite{Ref_R1973}. On the other hand, our proof does \textit{not} use a potential function argument. As we discuss at the end, the asymmetric and transitive relation we use is not even a potential function, i.e., the argumentation explained above does not work for this relation.

In our perspective, this aspect of our proof illustrates an interesting situation as follows. For a potential function argument to work, the game form at hand must possess the finite improvement property, i.e., any outcome of the game must converge to a stable outcome via natural game play \cite{Ref_MS1996}. However, it is practically not possible to show that a game form has the finite improvement property without first using a potential function argument. Hence, it is hard to know beforehand whether any potential function argument will ever be successful in obtaining an existence guarantee for a given game form. In this sense, our proof nicely demonstrates that even if an educated guess for a potential function fails for the game form at hand, that guess might be still instrumental for obtaining an existence guarantee; and hence, it should not be dismissed immediately.

\paragraph{Our Results.} A coalition of agents may deviate (i.e., they may jointly change their strategies) from an outcome of the game if by doing so they can attain a superior coalitional outcome in the Pareto sense, i.e., we allow for weakly improving deviations. Ideally, we would be interested in a stable outcome that is resilient to weakly improving deviations by any coalition of agents. An outcome satisfying this property is known in the literature as a super strong equilibrium \cite{Ref_FT2009}. However, this equilibrium notion does not exist in most game forms, including ISCGs \cite{Ref_FT2009}. Instead, we prove existence of outcomes satisfying the following three axioms simultaneously:

\begin{description}
\item[\textbf{(A1)}$\:\!$]\small Resilience to weakly improving deviations by singletons

$\quad$(i.e., the outcome should be a \textbf{Nash equilibrium}).

\item[\textbf{(A2)}$\:$] \small Resilience to weakly improving deviations by the grand coalition

$\quad$(i.e., the outcome should be \textbf{Pareto efficient}).

\item[\textbf{(A3)}$\:$] \small Resilience to weakly improving deviations by \emph{a priori} given set of coalitions,

$\quad\;$which partitions the set of agents

$\quad$(i.e., the outcome should be a \textbf{partition equilibrium}).
\end{description}

The axiom \textbf{A1}, requiring that an outcome be resilient to weakly improving deviations by singletons, is the simplest notion of stability that one may consider, and it is commonly used in game-theoretic studies. In a similar vein, the axiom \textbf{A2}, requiring that an outcome be Pareto efficient, is the most natural notion of efficiency, and it is commonly used in the economics literature. A few words are in order, however, pertaining to axiom \textbf{A3}.

Since a super strong equilibrium does not exist in most game forms, a growing trend in the recent literature is to study equilibrium outcomes under various restrictions on coalition formation. The kind of coalitions that agents may form can be specified in the form of an \emph{a priori} given coalition structure. In the context of \textit{resource selection games} (RSGs), Feldman and Tennenholtz \cite{Ref_FT2009} considered a partition coalition structure (i.e., the set of viable coalitions is a partition of the set of agents). And on the basis of a partition coalition structure, they introduced the notion of a partition equilibrium. Also related to this notion, a major theme in multi-agent systems is how autonomous agents may come together and form coherent groupings in order to pursue their individual or collective goals more effectively. In the literature, this problem is known as the coalition structure generation problem; see \cite{Ref_CSGPS}. How a coalition structure is formed is beyond the scope of our paper. We proceed with the assumption that a partition coalition structure is \emph{a priori} given.

In singleton congestion games where each cost function is chosen from three convex increasing functions, even when each resource is accessible by all agents, there may not exist an outcome that satisfies the axioms \textbf{A1}, \textbf{A2}, \textbf{A3}; see \cite{Arxiv}. However, we show (in Theorem \ref{TeoEx}) that in ISCGs there always exists an outcome that satisfies the above three axioms. In other words, we show that in ISCGs one need not sacrifice efficiency even if attention is confined to Nash and partition equilibrium outcomes.

\paragraph{Related Work:} The aforementioned RSGs are similar to ISCGs in that they also involve a set of agents selecting from a set of resources to utilize. But there are two differences: In RSGs, it is assumed that every agent has access to every resource. In this respect, ISCGs are more general than RSGs. On the other hand, the cost functions of resources need not be identical in RSGs. In this respect, RSGs are more general than ISCGs. Singleton congestion games generalize both these game forms in the sense that an agent need not have access to every resource, and the cost functions of resources need not be identical. Moreover, ISCGs are equivalent to the subclass of project games with identical rewards and identical agents' weights \cite{ProjectGames}.

Note that every finite game admits an outcome that satisfies axiom \textbf{A2} since the Pareto dominance relation is asymmetric and transitive. Also, in singleton congestion games an outcome that satisfies axiom \textbf{A1} always exists, since it is a subclass of congestion games, for which the existence of a Nash equilibrium is guaranteed \cite{Ref_R1973}. In RSGs with increasing cost functions, there always exists an outcome that satisfies axioms \textbf{A1} and \textbf{A3} \cite{Ref_ACH2013}.\footnote{It is easy to show that there also exists an outcome satisfying \textbf{A1} and \textbf{A2} in RSGs. However, it is open whether an outcome satisfying \textbf{A2} and \textbf{A3} always exists.} However, as mentioned earlier, there does not always exist an outcome that also satisfies axiom \textbf{A2} \cite{Arxiv}.

A stable outcome that is resilient to improving deviations (instead of weakly improving deviations) of any coalition is known in the literature as a strong equilibrium. It is known that a strong equilibrium always exists in singleton congestion games with monotone cost functions \cite{Ref_HL1997}.

\bigskip

The remainder of the paper is organized as follows: In Section \ref{sec:model}, we formally define ISCGs, and then we present three lemmas that become useful in showing our main result. In Section \ref{sec:mainResult}, we present our main result (Theorem \ref{TeoEx}), and then we discuss various aspects of its proof with some examples.


\section{The Model and the Preliminaries}
\label{sec:model}

An \textit{identical singleton congestion game} (ISCG) is a triplet $\left\langle N,M,f\right\rangle $ where:
\begin{itemize}
\item[--] $N=\left \{1,2,\cdots,n\right\}$ is the set of agents;
\item[--] $M=\left \{1,2,\cdots,m\right\}$ is the set of resources;
\item[--] $f=\left(f_{i}\right)_{i\in M}$ is a sequence (called the \textit{feasibility constraint}) such that $f_{i} \subseteq N$ for each $i \in M$; and ${\textstyle \bigcup \nolimits_{i\in M}}f_{i}=N$.
\end{itemize}
An \textit{allocation} (or an outcome) is an ordered sequence $a=\left(a_{i}\right)_{i \in M}$ such that $a_{i} \cap a_{\overline{i}}=\emptyset$ for all $i, \overline{i} \in M$ (where $i\neq \overline{i}$); and $a_{1}\cup \cdots \cup a_{m}=N$. An allocation $a$ is \textit{feasible} if for each $i$, $a_{i}\subseteq f_{i}$. Let $\mathcal{A}$ be the domain of allocations. Let $\mathcal{A}^{f}\subseteq A$ be the domain of feasible allocations.

\medskip

The interpretation of the game is as follows: Under an allocation $a$, the agents in $a_{i}$ are served by resource $i$. The congestion level at resource $i$ is proportional to its cardinality $\left \vert a_{i}\right \vert$. Agents try to avoid congested resources. We assume that each agent has access to at least one resource (hence, $\mathcal{A}^{f} \neq \emptyset$). This assumption is innocuous since there is no point in including an agent in the set $N$ if it cannot be assigned to any resource.

\medskip

A \textit{coalition} $c\subseteq N$ is a nonempty subset of agents. We say that coalition $c$ \textit{blocks} an allocation $a\in \mathcal{A}^{f}$ if there exists an allocation $\overline{a}\in \mathcal{A}^{f}$ such that:

\begin{itemize}
\item[--] for each resource $i\in M$, $a_{i}\smallsetminus c=\overline{a}%
_{i}\smallsetminus c$;
\item[--] for each $\left(  j,i_{1},i_{2}\right)  \in \left(  c\times M\times
M\right)  $ such that $j\in a_{i_{1}}$ and $j\in \overline{a}_{i_{2}}$,
$\left \vert \overline{a}_{i_{2}}\right \vert \leq \left \vert a_{i_{1}%
}\right \vert $;
\item[--] for some $\left(  j,i_{1},i_{2}\right)  \in \left(  c\times M\times
M\right)  $ such that $j\in a_{i_{1}}$ and $j\in \overline{a}_{i_{2}}$,
$\left \vert \overline{a}_{i_{2}}\right \vert <\left \vert a_{i_{1}}\right \vert $.
\end{itemize}
In simpler terms: Coalition $c$ blocks allocation $a$ if coalition members can change the resources that they use in some manner such that the coalition becomes better off in the Pareto sense. We refer to the allocation that results (above, $\overline{a}$) ``the allocation induced when $c$ blocks $a$''.

\medskip

Let $\mathcal{P}_{\geq1}(N)$ be the domain of coalitions, i.e., $\mathcal{P}_{\geq1}(N)=\mathcal{P}(N)\smallsetminus\{ \emptyset \}$, where $\mathcal{P}(N)$ is the power set of $N $. A \textit{coalition structure} $C\subseteq \mathcal{P}_{\geq1}(N)$ is a potential set for viable coalitions. We say that an allocation $a\in \mathcal{A}^{f}$ is $C$\textit{-stable} if there exists no $c\in C$ such that $c$ blocks $a$. Note that defined this way, a \textit{super strong equilibrium} is a $\mathcal{P}_{\geq1}(N)$-stable allocation. That is, a super strong equilibrium is a feasible allocation $a$ such that there exists no coalition that blocks $a$.

Though the notion of a super strong equilibrium is very appealing, it does not always exist even for the very restricted instances of ISCGs.\footnote{See \cite{Ref_FT2009} for a very simple ISCG instance (with only three agents and two resources) and for which a super strong equilibrium does not exist.} Therefore, we consider in this paper less demanding conditions. Specifically, we are interested in the existence of feasible allocations that are Pareto efficient, a Nash equilibrium, and a partition equilibrium. We define them next.

\medskip

Let $\mathcal{P}_{=1}(N)=\left\{\left \{ 1 \right \}, \left \{ 2 \right \}, \cdots, \left\{n\right\}\right \}$. Notice that under the coalition structure $\mathcal{P}_{=1}(N)$, the only viable coalitions are singletons. As stated using our notation: An allocation $a \in \mathcal{A}^{f}$ is a \textit{Nash equilibrium} if $a$ is $\mathcal{P}_{=1}(N)$-stable. Furthermore, an allocation $a \in \mathcal{A}^{f}$ is \textit{Pareto efficient} if it is $\{N\}$-stable (i.e., if the grand coalition does not block $a$).

We also consider the situation where the set of viable coalitions is a partition of the set of agents. Formally, a \textit{partition coalition structure} $C$ is such that $c \cap \overline{c} = \emptyset$ for all $c, \overline{c} \in C$ (where $c \neq \overline{c}$); and $\bigcup \nolimits_{c\in C} = N$. Given a partition coalition structure $C$, we refer to an allocation $a\in \mathcal{A}^{f}$ as a \textit{partition equilibrium} if $a$ is $C$-stable.

\medskip

The main result of our paper is that in an ISCG, for any given partition coalition structure $C$, there exists a feasible allocation $a\in \mathcal{A}^{f}$ such that $a$ is $\mathcal{P}_{=1}(N)$-stable, $\{N\}$-stable, and $C$-stable. That is, we show that there always exists an outcome that satisfies the three axioms given in Section 1. We present this result in Theorem \ref{TeoEx} in Section 3. The key to our proof of this result is what we call the \textquotedblleft kernel values\textquotedblright\ of an allocation $a$. In the remainder of this section, we define the kernel values, and we present three lemmas that become useful in showing Theorem \ref{TeoEx}. To ease understanding, we will illustrate the concepts that we introduce by referring to the following example.

\begin{example}
\label{ExKernel}In an ISCG with four resources and eight agents, consider the coalition $c=\left\{1,2,6\right\}$, and the allocation $a$ such that: $a_{1} = \left\{1,2\right\}, a_{2} = \left\{3,4,5\right\}, a_{3} = \left\{6,7,8\right\}, a_{4} = \emptyset$.\hfill $\Diamond$
\end{example}

\noindent -- Given an allocation $a$, let $\omega: M \rightarrow M$ be a bijection such that $|a_{\omega \left(1\right)}| \geq |a_{\omega \left(  2\right)}| \geq \cdots \geq |a_{\omega \left(m\right)}|$. That is, under $\omega$ resources are ordered, in order of the number of agents assigned to them under $a$. For instance, for $a$ in Example \ref{ExKernel}, $\omega$ may be as follows: $\omega(1) = 2, \omega(2) = 3, \omega(3) = 1, \omega(4) = 4$. Looking at $\omega$, note that under $a$, the cardinality of resource $2$ is maximal and the cardinality of resource $4$ is minimal.

\medskip

\noindent -- The \emph{kernel} of allocation $a$, denoted by $k\left(a\right)$, is the ordered list: $(|a_{\omega \left(1\right)}|,$ $|a_{\omega \left(2\right)}| \cdots, |a_{\omega\left(m\right)}|)$. For instance, for $a$ in Example \ref{ExKernel}, we have $k(a)=(3,3,2,0)$.

\medskip

\noindent -- Given an allocation $a$, for some coalition $c$ let $\omega^{c}:M\rightarrow M$ be a bijection such that $|c\cap a_{\omega^{c}\left(  1\right)  } |\geq|c\cap a_{\omega^{c}\left(  2\right)  }|\geq \cdots \geq|c\cap a_{\omega^{c}\left(  m\right)  }|$. That is, under $\omega^{c}$ resources are ordered, in order of how many agents in coalition $c$ are assigned to them under $a$. For instance, for $a$ and $c$ in Example \ref{ExKernel}, $\omega^{c}$ may be as follows: $\omega^{c}(1)=1,\omega^{c}(2)=3,\omega^{c}(3)=2,\omega^{c}(4)=4$. Looking at $\omega^{c}$, note that under $a$, it is resource $1$ which is assigned the maximum number of agents in coalition $c$, which is followed by resource $3$. The resources $2$ and $4$ are ordered at the end under $\omega^{c}$ since no agent in $c$ is assigned to these two resources under $a$.

\medskip

\noindent -- The $c$-\emph{kernel} of allocation $a$, denoted by $k(c,a)$, is the ordered list: $(|c\cap a_{\omega^{c}\left(1\right)  }|$, $|c\cap a_{\omega^{c}\left( 2\right)}|$,$\cdots$,$|c\cap a_{\omega^{c}\left(m\right)}|)$. For instance, for $a$ and $c$ in Example \ref{ExKernel}, we have $k(c,a)=(2,1,0,0)$.

\medskip

\noindent -- Given an allocation $a$, for coalition $c$, let $c_{1},\cdots,c_{n}$ be the partition of coalition $c$ such that for each $j\in c_{s}$, if $j \in a_{i}$ then $\left \vert a_{i}\right \vert = s$. That is, under $a$, the agents in $c_{s}$ are those coalition members that are assigned to resources whose cardinality is $s$. For instance, for $a$ and $c$ in Example \ref{ExKernel}, we have: $c_{1} = \emptyset, c_{2} = \{1,2\}, c_{3} = \{6\}, c_{4} = c_{5} = c_{6} = c_{7} = c_{8} = \emptyset$.

\medskip

\noindent -- The $c$-\emph{welfare-kernel} of allocation $a$, denoted by $w\left(c,a\right)$, is the ordered list: $\left(|c_{n}|,|c_{n-1}|,\cdots,|c_{1}|\right)$. For instance, for $a$ and $c$ in Example \ref{ExKernel}, we have $w\left(c,a\right)=(0,0,0,0,0,1,2,0)$. Looking at $w\left(c,a\right)$, we can say that under $a$, no agent in coalition $c$ is assigned to a resource whose cardinality is 8 or 7 or 6 or 5 or 4 or 1. Also, we see that one coalition member is assigned to a resource with cardinality 3, and two coalition members are assigned to resources with cardinality 2.

\medskip

Below we present three lemmas pertaining to kernel values. These lemmas become instrumental in proving Theorem \ref{TeoEx}. Before that, however, we introduce the notion of a ``chain'', which helps to simplify our exposition. Consider two allocations, say $a$ and $\overline{a}$. An $a\overline{a}$\textit{-chain}, represented as $i_{1}\rightarrow \left(j_{1}\right) \rightarrow i_{2}\rightarrow \left(j_{2}\right)  \rightarrow \cdots \rightarrow \left(j_{s-1}\right) \rightarrow i_{s}$ ($s\geq2$), refers to the situation such that:
\begin{itemize}
\item[--] $j_{1}\in a_{i_{1}}$, $j_{1}\in \overline{a}_{i_{2}}$, and $i_{1} \neq i_{2}$;

$\vdots$

\item[--] $j_{s-1}\in a_{i_{s-1}}$, $j_{s-1}\in \overline{a}_{i_{s}}$, and $i_{s-1}\neq i_{s}$.
\end{itemize}
In loose terms, an $a\overline{a}\textit{-chain}$ specifies how allocation $a$ can be transformed into allocation $\overline{a}$: The transformation involves moving agent $j_{1}$ from resource $i_{1}$ to $i_{2}$, agent $j_{2}$ from resource $i_{2}$ to $i_{3}$, and so on. To transform $a$ into $\overline{a}$, it may be necessary to move some other agents, too, since there may exist some other $a\overline{a}\textit{-chains}$. Above, we refer to $s$ as the \textit{length} of the $a\overline{a}$-chain. Note that the minimum length of an $a\overline{a}$-chain is 2. Also, note that there always exists an $a\overline{a}$-chain unless $a=\overline{a}$.

\medskip

We also introduce the ``chain addition operator'', denoted by $\oplus$. We write: $\overline{a} = a \oplus i_{1}\rightarrow \left(j_{1}\right)\rightarrow i_{2}\rightarrow \left(j_{2}\right)\rightarrow \cdots \rightarrow \left(j_{s-1}\right) \rightarrow i_{s}$ if
\begin{itemize}
\item[--] $i_{1}\rightarrow \left(j_{1}\right)\rightarrow i_{2} \rightarrow \left(j_{2}\right)\rightarrow \cdots \rightarrow \left(j_{s-1}\right)\rightarrow i_{s}$ is an $a\overline{a}$-chain,
\item[--] for each $j\in N\smallsetminus \left\{j_{1}, j_{2}, \cdots, j_{s-1}\right\}$, the resource to which $j$ is assigned is the same under $a$ and $\overline{a}$.
\end{itemize}
In loose terms, the allocation $\overline{a}$ is obtained from $a$ if the resources that agents $j_{1},j_{2},\cdots,j_{s-1}$ use change as indicated by the $a\overline{a}$-chain. Notice that if $\overline{a}=a\oplus i_{1} \rightarrow \left(j_{1}\right) \rightarrow i_{2}\rightarrow \left(j_{2}\right) \rightarrow i_{3}$, then $\overline{a}=\left(a \oplus i_{1} \rightarrow \left(j_{1}\right) \rightarrow i_{2}\right) \oplus i_{2} \rightarrow \left(j_{2}\right) \rightarrow i_{3}$. Also, notice that if $i_{1}\rightarrow \left(j_{1}\right)  \rightarrow i_{2}$ is an $a\overline{a}$-chain, this does not imply that $\overline{a}=a\oplus i_{1} \rightarrow \left(  j_{1}\right)  \rightarrow i_{2}$ because the transformation of $a$ into $\overline{a}$ may require the execution of some other $a\overline{a}$-chains, too (besides $i_{1}\rightarrow \left(j_{1}\right) \rightarrow i_{2}$).

\medskip

Lemma \ref{LemSwitch} below states that given a partition coalition structure $C$, under some allocation if two agents in some coalition $c \in C$ switch their positions, then the kernel values remain the same as before. We delegate the easy proof to the Appendix.

\begin{lemma}
\label{LemSwitch}Let $C$ be a partition coalition structure. Let $j_{1},j_{2}\in c\in C$. Let allocations $a,\overline{a} \in \mathcal{A}$ be such that $\overline{a}= a \oplus i_{1}\rightarrow \left(j_{1}\right)\rightarrow i_{2} \rightarrow \left(j_{2}\right)\rightarrow i_{1}$. Then, the kernel values of $a$ and $\overline{a}$ are the same. Also, for each $\widetilde{c}\in C$, the $\widetilde{c}$-kernel and $\widetilde{c}$-welfare-kernel values of $a$ and $\overline{a}$ are the same.
\end{lemma}

We now introduce two symmetric and transitive relations that will become essential to prove Theorem \ref{TeoEx}. These relations are composed of lexicographical comparisons of various kernel values defined above. We will be using $\prec$ to denote ``lexicographically smaller than". The first one of these relations is defined for partition coalition structures. Given a partition coalition structure $C$, we say that allocation $a$
$C$-\textit{balance dominates} $\overline{a}$,
\begin{itemize}
\item[--] if $k\left(a\right)\prec k\left(\overline{a}\right)$;
\item[--] or if $k\left(a\right) = k\left(\overline{a}\right)$ and
\begin{itemize}
\item[-] for each $c\in C$, $k\left(c,a\right)\prec k\left(c,\overline{a}\right)$ or $k\left(c,a\right) = k\left(c,\overline{a}\right)$,
\item[-] and for some $c\in C$, $k\left(c,a\right)\prec k\left(c,\overline{a}\right)$.
\end{itemize}
\end{itemize}
In loose terms: Allocation $a$ $C$-balance dominates $\overline{a}$ if under $a$ the distribution of agents to resources is more even. Moreover, if there is a tie in this regard, then $a$ $C$-balance dominates $\overline{a}$ if the distribution of coalition members to resources is more even under $a$. Lemma \ref{LemCdom} below is a straightforward observation pertaining to $C$-balance dominance relation. We delegate the easy proof to the Appendix.

\begin{lemma} \label{LemCdom}
Let $C$ be a partition coalition structure. Let $j_{1} \in c \in C$. Let allocations $a$ and $\overline{a}$ be such that $\overline{a} = a\oplus i_{1} \rightarrow \left(j_{1}\right) \rightarrow i_{2}$. Then allocation $\overline{a}$ $C$-balance dominates $a$: $(a)$ if $\left\vert a_{i_{1}}\right\vert \geq \left \vert a_{i_{2}}\right\vert + 2$,
$(b)$ or if $\left\vert a_{i_{1}}\right \vert = \left \vert a_{i_{2}}\right\vert + 1$ and $\left\vert c \cap a_{i_{1}}\right\vert \geq \left\vert c \cap a_{i_{2}}\right\vert + 2$.
\end{lemma}

The second relation that we introduce compares $c$-welfare-kernel values of allocations, again lexicographically. For a coalition $c$, we say that allocation $a$ $c$-\textit{welfare-dominates} $\overline{a}$ if $w\left(c,a\right) \prec w\left(c,\overline{a}\right)$. Lemma \ref{Lemcwelfdom} pertains to this relation, whose proof we delegate to the Appendix.

\begin{lemma} \label{Lemcwelfdom}
Let $C$ be a partition coalition structure. Let $c\in C$ be a coalition. Let $a$ be an allocation.

\medskip

\noindent (a) Suppose that coalition $c$ blocks allocation $a$. Let
$\widetilde{a}$ be the allocation induced when $c$ blocks $a$. Then,
$\widetilde{a}$ $c$-welfare-dominates $a$.

\medskip

\noindent (b) Suppose that allocation $\overline{a}$ $c$-welfare-dominates $a$.
Also, suppose that $\widetilde{M}\subset M$ is such that for each
$s\in \{1,2,\cdots,n\}$,
\[{\textstyle \sum \nolimits_{i\in \widetilde{M},\left \vert a_{i}\right \vert =s}}
\left \vert c\cap a_{i}\right \vert =%
{\textstyle \sum \nolimits_{i\in \widetilde{M},\left \vert \overline{a}%
_{i}\right \vert =s}}
\left \vert c\cap \overline{a}_{i}\right \vert .\]
Let $k_{max}=max_{i\in M\smallsetminus \widetilde{M}}\left \vert a_{i}\right \vert $ and $\overline{k}_{max}=max_{i\in M\smallsetminus \widetilde{M}}\left \vert \overline{a}_{i}\right \vert $. Then,
\[
\overline{k}_{max}\leq k_{max}, \text{and }
{\textstyle \sum \nolimits_{i\in M\smallsetminus \widetilde{M},\left \vert
a_{i}\right \vert =k_{max}}}
\left \vert c\cap a_{i}\right \vert \geq
{\textstyle \sum \nolimits_{i\in M\smallsetminus \widetilde{M},\left \vert
\overline{a}_{i}\right \vert =k_{max}}}
\left \vert c\cap \overline{a}_{i}\right \vert
\]
\end{lemma}

\section{The Main Result}
\label{sec:mainResult}

This section is devoted to our main result: In Theorem \ref{TeoEx}, we show that in an ISCG, there always exists an allocation that satisfies the three axioms given in Section 1. The proof of the theorem is fairly involved. We prove the theorem by showing that every maximal outcome with respect to the $C$-balance dominance relation satisfies the three axioms. At the end of the section, we also present an example and show that it is possible that an outcome satisfies the three axioms and yet not be maximal with respect to the $C$-balance dominance relation.

\begin{theorem}\label{TeoEx}
In an ISCG, for any given partition coalition structure $C$, there always exists an allocation $a\in A^{f}$ such that $a$ is $\mathcal{P}_{=1}(N)$-stable, $\{N\}$-stable, and $C$-stable. That is, in an ISCG there always exists a Pareto efficient outcome that is a Nash equilibrium and a partition equilibrium.
\end{theorem}

\begin{proof}
Let $C$ be a given partition coalition structure. Let $\kappa \left(C,A^{f}\right)\subseteq A^{f}$ be such that for each $a\in \kappa \left(C,A^{f}\right)$, there exists no $\overline{a}\in A^{f}$ such that $\overline{a}$ $C$-balance dominates $a$. That is, $\kappa \left(C,A^{f}\right)$ is the set of maximal allocations in $A^{f}$ with respect to the $C$-balance dominance relation. Note that $\kappa \left(C,A^{f}\right) \neq \emptyset$ since the $C$-balance dominance relation is transitive and asymmetric. Let $a\in \kappa \left(C,A^{f}\right)$. We prove the theorem by showing that $a$ is $\mathcal{P}_{=1}(N)$-stable (\textbf{A1}), $\{N\}$-stable (\textbf{A2}), and $C$-stable (\textbf{A3}); the first two of which are already known due to \cite{HARKS} as the first component of $C$-balance dominance relation lexicographically compares the sorted vector of cardinalities of resources.\footnote{For the sake of completeness, the proofs of that $a$ satisfies axioms \textbf{A1} and \textbf{A2} are also given in the Appendix.} Hence, we only need to show that $a$ also satisfies axiom \textbf{A3}.

\medskip

By way of contradiction, suppose that $a$ is not $C$-stable. Then, there exists $c \in C$ such that $c$ blocks $a$. Let $\overline{a}\in A^{f}$ be the allocation induced when $c$ blocks $a$. Then, by Lemma \ref{Lemcwelfdom}(a), we obtain that $\overline{a}$ $c$-welfare-dominates $a$. Thus, the following two statements are true:

\begin{description}
\item[\textbf{(S1)}] $a\in \kappa \left(C,A^{f}\right)$,

\item[\textbf{(S2)}] $\overline{a}\in A^{f}$ and $\overline{a}$ $c$-welfare-dominates $a$.
\end{description}

We will prove that $a$ is $C$-stable by showing the supposition that \textbf{S1} and \textbf{S2} are true leads to a contradiction. The outline of the proof is as follows: Below, in Steps 1 and 2 we update $\overline{a}$ and $A^{f}$ iteratively. We show that after each update \textbf{S1} and \textbf{S2} remain to be true. And once we are done with our updates we derive a contradiction.

\medskip

\noindent \textbf{Step 1:} Update the feasibility constraints as follows:

\begin{itemize}
\item[-] for each $i \in M$, let $f_{i} = a_{i}\cup \overline{a}_{i}$.
\end{itemize}
	
\noindent \textbf{Step 2:} As long as there exists an $a\overline{a}$-chain whose length is bigger than 2, represented as
$$i_{1}\rightarrow \left(  j_{1}\right)  \rightarrow i_{2}\rightarrow \left(j_{2}\right)  \rightarrow i_{3}\rightarrow \cdots \left(  j_{s-1}\right) \rightarrow i_{s},$$
update the feasibility constraint $f_{i_{3}}$ and the allocation $\overline{a}$ as follows:

\begin{itemize}
\item[--] $f_{i_{3}}:=f_{i_{3}}\smallsetminus \left \{  j_{2}\right \}  \cup \left \{j_{1}\right \}  $;
\item[--] $\overline{a}:=\overline{a}\oplus i_{3}\rightarrow \left(  j_{2}\right) \rightarrow i_{2}\rightarrow \left(  j_{1}\right)  \rightarrow i_{3}$.
\end{itemize}

\noindent\underline{Consider the update at Step 1:} After the update note that the set $A^{f}$ becomes smaller but we still have $a,\overline{a}\in A^{f}$. Hence, \textbf{S1} and \textbf{S2} remain to be true.

\medskip

\noindent\underline{Consider an update at Step 2:} After the update note that $\overline{a}$ and $A^{f}$ change but we still have $a,\overline{a}\in A^{f}$. Also, by Lemma \ref{LemSwitch}, for $\overline{a}$ the kernel values remain the same as before. Hence, after the update \textbf{S2} remains to be true.

Let $\widetilde{a}$ be an allocation that was infeasible before the update but that becomes feasible after the update. It is clear that under $\widetilde{a}$, $j_{1}\in \widetilde{a}_{i_{3}}$ and $j_{2}\in \widetilde{a}_{i_{2}}$. To show that after the update \textbf{S1} remains to be true, we need to show that $\widetilde{a}$ does not $C$-balance dominate $a$. Let $\widehat{a}= \widetilde{a} \oplus i_{3}\rightarrow \left(j_{1}\right) \rightarrow i_{2}\rightarrow \left(  j_{2}\right)  \rightarrow i_{3}$. Note that before the update $\widehat{a}$ was feasible. Thus, $\widehat{a}$ does not $C$-balance dominate $a$. But by Lemma \ref{LemSwitch}, for $\widetilde{a}$ and $\widehat{a}$ the kernel values are the same. Then, $\widetilde{a}$ does not $C$-balance dominate $a$ either. Hence, after the update \textbf{S1} remains to be true.

\medskip

Also, note that at Step 2 each time $\overline{a}$ and $A^{f}$ are updated, the length of an $a\overline{a}$-chain becomes smaller. Therefore, Step 2 terminates after a finite number of iterations. And when it terminates, we obtain that:

\begin{itemize}
\item[--] for $a$, $\overline{a}$, $A^{f}$, the statements \textbf{S1} and \textbf{S2} are true;
\item[--] every $a\overline{a}$-chain that remains is of length 2.
\end{itemize}

We are now ready to derive a contradiction. Since every $a\overline{a}$-chain that remains is of length 2, we can divide the set $M$ into the following three subsets:
$$M^{0}=\left\{i\in M|a_{i}=\overline{a}_{i}\right\}, \quad
M^{-}=\left \{i\in M|\overline{a}_{i}\subset a_{i}\right\}, \quad
M^{+}=\left \{i\in M|\overline{a}_{i}\supset a_{i}\right\}.$$
Note that for $a$ and $\overline{a}$, each $a\overline{a}$-chain is of the form $i_{1}\rightarrow \left(j\right) \rightarrow i_{2}$, where $i_{1} \in M^{-}$ and $i_{2}\in M^{+}$. We also divide the sets $M^{-}$ and $M^{+}$ into their partitions with respect to the cardinalities of resources under $a$: Let $M^{-}(k) = \left\{ i\in M^{-} | |a_{i}| = k \right\}$ and $M^{+}(k) = \left\{ i \in M^{+} | |a_{i}| = k\right\}$.

\medskip

To derive a contradiction we use Lemma \ref{Lemcwelfdom}(b). We proceed as follows: First, we identify a set $\widetilde{M}\subset M$ for which the suppositions in Lemma \ref{Lemcwelfdom}(b) are satisfied. Then, using some arguments (below, under the heading \textbf{Iteration}), we show that the set $\widetilde{M}$ can be updated iteratively such that after each iteration its cardinality increases and yet the suppositions in Lemma \ref{Lemcwelfdom}(b) are still satisfied. Eventually, we obtain that $\widetilde{M} = M$ and for each $s \in \left\{1,2,\cdots,n\right\}$,
$$\textstyle \sum \nolimits_{i \in M, |a_{i}| = s} |c \cap a_{i}| = \sum \nolimits_{i \in M, |\overline{a}_{i}| = s} |c \cap \overline{a}_{i}|.$$
This will imply that for $a$ and $\overline{a}$, the $c$-welfare-kernel values are indeed the same, which will contradict that $\overline{a}$ $c$-welfare-dominates $a$ (\textbf{S2}) and complete our proof.

\medskip

Initially, we set $\widetilde{M} = M^{0}$.

\medskip

For $\widetilde{M}=M^{0}$, it is easy to verify that the suppositions in Lemma \ref{Lemcwelfdom}(b) are satisfied. Let $k_{max}=max_{i \in M \smallsetminus \widetilde{M}}\left\vert a_{i} \right\vert$ and $\overline{k}_{max}=max_{i \in M \smallsetminus \widetilde{M}} \left\vert \overline{a}_{i} \right\vert $. Then, by Lemma \ref{Lemcwelfdom}(b) we have $\overline{k}_{max} \leq k_{max}$, and
\begin{equation}
\textstyle \sum \nolimits_{i \in M \smallsetminus \widetilde{M},
\left\vert a_{i} \right\vert = k_{max}}
\left\vert c \cap a_{i} \right\vert \geq
\sum \nolimits_{i \in M \smallsetminus \widetilde{M},
\left\vert \overline{a}_{i} \right\vert = k_{max}}
\left\vert c \cap \overline{a}_{i} \right\vert. \tag{$\ast$}
\end{equation}

\textbf{Iteration:}
Consider $M^{+}\left(k\right)$ where $k \geq k_{max}$. Suppose that $M^{+}\left(k\right) \neq \emptyset$. Let $i\in M^{+}\left(k\right)$. Since $\left\vert a_{i} \right\vert = k$ and $\overline{a}_{i} \supset a_{i}$, we obtain that $\left \vert \overline{a}_{i}\right \vert \geq k+1\geq k_{max}+1$. But this contradicts that $\overline{k}_{max}\leq k_{max}$. Thus, for $k \geq k_{max}$, $M^{+}\left(  k\right) = \emptyset$.

We now restrict our attention to the sets $M^{-}\left(  k_{max}\right)  $ and $M^{+}\left(  k_{max}-1\right)  $: Let $i_{1}\in M^{-}\left(  k_{max}\right)$. Consider an $a\overline{a}$-chain of the form $i_{1}\rightarrow \left(j_{1}\right)  \rightarrow i_{2}$. Let $\widetilde{a}=a\oplus i_{1} \rightarrow \left(  j_{1}\right)  \rightarrow i_{2}$. Note that $\widetilde{a}\in A^{f}$. By Lemma \ref{LemCdom}, if $\left \vert a_{i_{2}}\right \vert\leq k_{max}-2$, or if $\left \vert a_{i_{2}}\right \vert =k_{max}-1$ and $\left \vert c\cap a_{i_{1}}\right \vert \geq \left \vert c\cap a_{i_{2}}\right \vert +2$, we obtain that $\widetilde{a}$ $C$-balance dominates $a$. But this contradicts that $a\in \kappa \left(C,A^{f}\right)$. Thus, $i_{2}\in M^{+}\left(  k_{max}-1\right)  $ and $\left \vert c\cap a_{i_{1}}\right \vert \leq \left \vert c\cap a_{i_{2}}\right \vert +1$. Suppose that there exists an $a\overline{a}$-chain of the form $i\rightarrow \left(j\right) \rightarrow i_{2}$ where $j\neq j_{1}$. Then, $a_{i_{2}}\cup \left \{j,j_{1}\right\} \subseteq \overline{a}_{i_{2}}$. Then, $\left \vert \overline{a}_{i_{2}}\right \vert \geq \left \vert a_{i_{2}}\right \vert +2=k_{max}+1$. But this contradicts that $\overline{k}_{max}\leq k_{max}$. Thus, there exists a one-to-one function $\upsilon:M^{-}\left(k_{max}\right) \rightarrow M^{+}\left(  k_{max}-1\right)$ such that:
\begin{itemize}
\item[--] for each $i\in M^{-}\left(  k_{max}\right)$, there exists an $a\overline{a}$-chain of the form $i\rightarrow \left(j\right)  \rightarrow \upsilon\left(i\right)$.
\end{itemize}
And for $\upsilon$, as argued above, the following holds:
\begin{itemize}
\item[--] for each $i \in M^{-} (k_{max})$, $|c \cap \overline{a}_{\upsilon(i)}| = |c \cap a_{\upsilon (i)}| + 1 \geq |c \cap a_{i}|$.
\end{itemize}
But then the inequality above in ($\ast$) holds only if $\upsilon$ is a bijection (i.e., $|M^{-}(k_{max})|$ $=$ $|M^{+}(  k_{max}-1)|$) and for each $i \in M^{-} \left(k_{max}\right)$, $|c \cap \overline{a}_{\upsilon(i)}| = |c \cap a_{i}|$. Therefore,
$$\textstyle \sum \nolimits_{i \in M \smallsetminus \widetilde{M}, |a_{i}| = k_{max}} |c \cap a_{i}| = \sum \nolimits_{i \in M \smallsetminus \widetilde{M}, |\overline{a}_{i}| = k_{max}} |c \cap \overline{a}_{i}|.$$
But then if we update $\widetilde{M}$ and set $\widetilde{M} := \widetilde{M} \cup M^{-}(k_{max}) \cup M^{+}(k_{max}-1)$, the suppositions in Lemma \ref{Lemcwelfdom}(b) are still satisfied.

\medskip

Therefore, as argued above, we can iterate these arguments and update $\widetilde{M}$ until we obtain that $\widetilde{M}=M$. And then we can conclude that for both $a$ and $\overline{a}$, the $c$-welfare-kernel values are actually the same. However, this contradicts that $\overline{a}$ $c$-welfare-dominates $a$ (\textbf{S2}). Therefore, our initial supposition must be wrong, i.e., allocation $a$ must be $C$-stable. This completes our proof. \qed
\end{proof}

In the proof of Theorem \ref{TeoEx}, we showed that if $a \in \kappa(C, A^f)$, then $a$ satisfies the axioms \textbf{A1}, \textbf{A2}, and \textbf{A3}, given in Section 1. The following example shows that the converse of this statement is not true.

\begin{example}
Consider an ISCG with two resources and fifteen agents such that both resources are accessible to every agent. Consider the following partition coalition structure $C=\{\{1,2,3\}, \{4,5,6\}, \{7,8,9\}, \{10,11,12\}, \{13,14,15\}\}$. Let allocation $a$ be as follows: $a_1 = \{1, 2, 4, 5, 7, 8, 10, 11\}$, $a_2 = \{3, 6, 9, 12, 13, 14, 15\}$. Also, let allocation $\overline{a}$ be as follows: $\overline{a}_1 = \{1, 2, 4, 7, 8, 10, 13, 14\}$, $\overline{a}_2 = \{3, 5, 6, 9,$ $11, 12, 15\}$. It is easy to verify that the allocation $a$ is $\mathcal{P}_{=1}(N)$-stable, $\{N\}$-stable, and $C$-stable. However, $a \notin \kappa(C, A^f)$ since $\overline{a}$ $C$-balance dominates $a$. \hfill $\Diamond$
\end{example}

As a technical note, the following example shows that $C$-balance dominance relation is not a potential function, as we claimed at the beginning of the paper. That is, an allocation $\bar{a}$ induced when a coalition $c \in C$ blocks an allocation $a$, does not necessarily $C$-balance dominate $a$.

\begin{example}
Consider an ISCG with four resources and eighteen agents such that all resources are accessible to every agent. Consider the following partition coalition structure $C = \{\{1, \ldots, 13\}, \{14, \ldots, 18\}\}$. Let allocation $a$ be as follows: $a_1 = \{1, 14, 15, 16\}$, $a_2 = \{2, 3, 4, 5\}$, $a_3 = \{6, 7, 8, 9, 17\}$, $a_4 = \{10, 11, 12,$ $13, 18\}$. Notice that coalition $c = \{1, \ldots, 13\}$ blocks allocation $a$ where the induced allocation $\bar{a}$ is as follows:
$\bar{a}_1 = \{6, 10, 14, 15, 16\}$, $\bar{a}_2 = \{7, 8, 9, 11, 12\}$, $\bar{a}_3 = \{1, 2, 3, 17\}$, $\bar{a}_4 = \{4, 5, 13, 18\}$. Note that agent $13$ gets better off whereas no agent in coalition $c$ gets worse off under $\bar{a}$. However, notice that  allocation $\bar{a}$ does not $C$-balance dominate allocation $a$. This is because $k(a) = k(\bar{a})$ yet $k(c,a) \not\prec k(c,\bar{a})$ since $k(c,a) = (4, 4, 4, 1)$ and $k(c, \bar{a}) = (5, 3, 3, 2)$. \hfill $\Diamond$
\end{example}

It is also natural to ask whether Theorem \ref{TeoEx} can be extended in way that also incorporates overlapping coalitions, instead of a partition coalition structure as in axiom \textbf{A3}. The following example shows that this is not possible even in very restricted settings.

\begin{example}
Consider an ISCG with two resources and three agents such that agent 1 can only access to resource 1, agent 2 can only access to resource 2, and agent 3 can access to both resources. Consider the following coalition structure $C = \{\{1, 3\}, \{2, 3\}\}$. Notice that there are two possible allocation both of which are not $C$-stable. \hfill $\Diamond$
\end{example}

\bibliographystyle{splncs04}
\bibliography{references}

\begin{thebibliography}{10}
\providecommand{\url}[1]{\texttt{#1}}
\providecommand{\urlprefix}{URL }
\providecommand{\doi}[1]{https://doi.org/#1}

\bibitem{Ref_ACH2013}
Anshelevich, E., Caskurlu, B., Hate, A.: Partition equilibrium always exists in
  resource selection games. Theory of Computing Systems  \textbf{53}(1),
  73--85 (apr 2013)

\bibitem{CollisionAvoidance}
Berg, J.v.d., Guy, S.J., Lin, M., Manocha, D.: Reciprocal n-body collision
  avoidance. In: Pradalier, C., Siegwart, R., Hirzinger, G. (eds.) Robotics
  Research. pp. 3--19. Springer Berlin Heidelberg, Berlin, Heidelberg (2011)

\bibitem{ProjectGames}
Bilo, V., Gourves, L., Monnot, J.: Projectgames. In: Algorithms and Complexity
  (CIAC 2019). vol. 11485, pp. 75--86. Springer International Publishing (2019)

\bibitem{Arxiv}
{Caskurlu}, B., {Ekici}, O., {Erdem Kizilkaya}, F.: {On Existence of
  Equilibrium Under Social Coalition Structures}. arXiv e-prints
  arXiv:1910.04648 (Oct 2019)

\bibitem{Ref_FT2009}
Feldman, M., Tennenholtz, M.: Structured coalitions in resource selection
  games. {ACM} Transactions on Intelligent Systems and Technology
  \textbf{1}(1),  1--21 (oct 2010)

\bibitem{HARKS}
Harks, T., Klimm, M., Mohring, R.H.: Strong nash equilibria in games with the
  lexicographical improvement property. In: Proceedings of the 5th
  International Workshop on Internet and Network Economics. p. 463–470.
  WINE’09, Springer-Verlag, Berlin, Heidelberg (2009)

\bibitem{Ref_HL1997}
Holzman, R., Law-Yone, N.: Strong equilibrium in congestion games. Games and
  Economic Behavior  \textbf{21}(1-2),  85--101 (oct 1997)

\bibitem{SCG}
Ieong, S., McGrew, R., Nudelman, E., Shoham, Y., Sun, Q.: Fast and compact: A
  simple class of congestion games. In: Proceedings of the 20th National
  Conference on Artificial Intelligence - Volume 2. p. 489–494. AAAI’05,
  AAAI Press (2005)

\bibitem{Ref_MS1996}
Monderer, D., Shapley, L.S.: Potential games. Games and Economic Behavior
  \textbf{14}(1),  124--143 (may 1996)

\bibitem{Ref_CSGPS}
Rahwan, T., Michalak, T.P., Wooldridge, M., Jennings, N.R.: Coalition structure
  generation: A survey. Artificial Intelligence  \textbf{229},  139--174 (dec
  2015)

\bibitem{Ref_R1973}
Rosenthal, R.W.: A class of games possessing pure-strategy nash equilibria.
  International Journal of Game Theory  \textbf{2}(1),  65--67 (dec 1973)

\bibitem{Wardrop}
Wardrop, J.G.: Some theoretical aspects of road traffic research. Proceedings
  of the institution of civil engineers  \textbf{1}(3),  325--362 (1952)

\end{thebibliography}

\newpage

\appendix
\section*{Appendix}

\noindent\underline{\textit{Missing Proof of Lemma \ref{LemSwitch}:}}

\begin{proof}
Note that $\overline{a}$ results from $a$ when $j_{1},j_{2}\in c$ switch their
positions under $a$. But when two agents switch positions, no change ensues in
regard to how many agents are assigned to which resource. Thus, by definition,
for $a$ and $\overline{a}$, the kernel values are the same.

About the $\widetilde{c}$-kernels and $\widetilde{c}$-welfare-kernels of
allocations $a$ and $\overline{a}$: Note that their values depend only on how
many agents are assigned to each resource and from which coalition. But that
being fixed, their values do not depend on how members of a coalition are
allocated to the slots designated for that coalition. In other words, at some
allocation, the $\widetilde{c}$-kernel and $\widetilde{c}$-welfare-kernel
values would not change if the labels of members of a coalition were
reshuffled. Thus, when $j_{1},j_{2}\in c$ switch positions, the $\widetilde
{c}$-kernel and $\widetilde{c}$-welfare-kernel values remain the same as before. \qed
\end{proof}

\noindent\underline{\textit{Missing Proof of Lemma \ref{LemCdom}:}}

\begin{proof}
(a) Let $\left \vert a_{i_{1}}\right \vert \geq \left \vert a_{i_{2}}\right \vert
+2$. But then it is clear that the distribution of agents to resources becomes
more even when $j_{1}$ moves from $i_{1}$ to $i_{2}$. Thus, allocation
$\overline{a}$ $C$-balance dominates $a$.

(b) Let $\left \vert a_{i_{1}}\right \vert =\left \vert a_{i_{2}}\right \vert +1$
and $\left \vert c\cap a_{i_{1}}\right \vert \geq \left \vert c\cap a_{i_{2}%
}\right \vert +2$. When $j_{1}$ moves from $i_{1}$ to $i_{2}$, it is clear that
there will not be a change in the kernel value. Thus, $a$ and $\overline{a}$,
cannot be compared with respect to the $C$-balance dominance relation
by looking at the kernel values of $a$ and $\overline{a}$. We thus compare them
looking at the second criterion: How \textquotedblleft
balanced\textquotedblright \ is the distribution of coalition members to
resources under $a$ and $\overline{a}$.

For a coalition $\overline{c}\in C$, $\overline{c}\neq c$, note
that there will not be a change in the distribution of members of this
coalition to resources when $j_{1}\notin \overline{c}$ moves from $i_{1}$ to
$i_{2}$. And for coalition $c$, since $\left \vert c\cap a_{i_{1}}\right \vert
\geq \left \vert c\cap a_{i_{2}}\right \vert +2$, it is clear that the move of
$j_{1}\in c$ from $i_{1}$ to $i_{2}$ leads to a more even distribution of
members of coalition $c$ to resources. Thus, $\overline{a}$ is more
\textquotedblleft balanced\textquotedblright \ than $a$ according to the second
criterion. Thus, allocation $\overline{a}$ $C$-balance dominates $a$. \qed
\end{proof}

\noindent\underline{\textit{Missing proof of Lemma \ref{Lemcwelfdom}:}}

\begin{proof}
(a) Let $\left(  x_{1},x_{2},\cdots,x_{n}\right)  $ and $\left(  y_{1}%
,y_{2},\cdots,y_{n}\right)  $ be, respectively, the $c$-welfare-kernel values
of allocations $a$ and $\widetilde{a}$.

\medskip

If $y_{1}<x_{1}$, $\widetilde{a}$ $c$-welfare-dominates $a$ (by definition)
and we are done. Thus, suppose that $y_{1}\geq x_{1}$.

\medskip

If $y_{1}>x_{1}$, then
there are more agents in $c$ assigned to an $n$-cardinality resource under
$\widetilde{a}$. But then at least one agent in $c$ must be worse off under
$\widetilde{a}$. This contradicts that $\widetilde{a}$ is the allocation
induced when $c$ blocks $a$. Thus, we must have $y_{1}=x_{1}$.

\medskip

Given that $y_{1}=x_{1}$, if $y_{2}<x_{2}$, $\widetilde{a}$ $c$%
-welfare-dominates $a$ (by definition) and we are done. Thus, suppose that
$y_{2}\geq x_{2}$. If $y_{2}>x_{2}$, we get $y_{1}+y_{2}>x_{1}+x_{2}$. But
then there are more agents in $c$ assigned to a resource with cardinality $n$
or $n-1$ under $\widetilde{a}$. But then at least one agent in $c$ must be
worse off under $\widetilde{a}$. This contradicts that $\widetilde{a}$ is the
allocation induced when $c$ blocks $a$. Thus, we must have $y_{2}=x_{2}$.

\medskip

The above arguments can be iterated. Eventually, we obtain that either
$\widetilde{a}$ $c$-welfare-dominates $a$, or $\left(  x_{1},x_{2}%
,\cdots,x_{n}\right)  =\left(  y_{1},y_{2},\cdots,y_{n}\right)  $. But if this
latter case were true, it would be impossible for a member of coalition
$c$ to be better off under $\widetilde{a}$, without another coalition member
becoming worse off. This contradicts that $\widetilde{a}$ is the allocation
induced when $c$ blocks $a$. Thus, we conclude that $\widetilde{a}$
$c$-welfare-dominates $a$.

\medskip

(b) We are given that $\overline{a}$ $c$-welfare-dominates $a$. Note that we
can think of the $c$-welfare-dominance comparison of $\overline{a}$ and $a$
as the sum of two parts: A $c$-welfare-dominance comparison of them when
attention is confined to resources in $\widetilde{M}$. And a $c$%
-welfare-dominance comparison of them when attention is confined to resources
in $M\smallsetminus \widetilde{M}$.

\medskip

The former comparison leads to a tie (due to the equality in the lemma statement).

\medskip

About the latter comparison:

\medskip

Suppose that $\overline{k}_{max}>k_{max}$. Let
$i^{\ast}\in M\smallsetminus \widetilde{M}$ be a $\overline{k}_{max}%
$-cardinality resource under $\overline{a}$. Note that resource $i^{\ast}$'s
cardinality increases when we go from $a$ to $\overline{a}$. Since the only
agents that move are those in $c$, under $\overline{a}$ there is an agent in
$c$ who is assigned to $i^{\ast}$. Thus, under $\overline{a}$, there is at
least one agent in $c$ that is assigned to a $\overline{k}_{max}$-cardinality
resource in $M\smallsetminus \widetilde{M}$, but under $a$, all agents in $c$
that are assigned to resources in $M\smallsetminus \widetilde{M}$ are assigned
to resources whose cardinality is less than $\overline{k}_{max}$. But then it
is clear that, in the $c$-welfare-dominance comparison, when attention is
confined to resources in $M\smallsetminus \widetilde{M}$, $a$ comes superior.
This contradicts that $\overline{a}$ $c$-welfare-dominates $a$. Thus,
$\overline{k}_{max}\leq k_{max}$.

\medskip

Suppose that $\overline{k}_{max}<k_{max}$. But then the inequality in the lemma
statement holds trivially and the proof is complete. Thus, suppose that
$\overline{k}_{max}=k_{max}$.

\medskip

Suppose that the inequality in the lemma statement does not hold.
But then it is clear that, in the $c$-welfare-dominance comparison, when attention
is confined to resources in $M\smallsetminus \widetilde{M}$, $a$ comes superior.
This contradicts that $\overline{a}$ $c$-welfare-dominates $a$. Therefore,
the inequality also holds when $\overline{k}_{max}=k_{max}$. And this completes our proof.
\qed
\end{proof}

\noindent\underline{\textit{Missing proof from Theorem \ref{TeoEx} that $a$ satisfies axiom} \textbf{A1} \textit{(Harks et al. \cite{HARKS}):}}

\begin{proof}
Suppose that $a$ is not a Nash equilibrium. Without loss of generality, let it be such that $\left\{j_{1}\right\}$ blocks $a$, $\overline{a}\in A^{f}$ is the allocation induced, and $\overline{a} = a \oplus i_{1}\rightarrow \left(  j_{1}\right)\rightarrow i_{2}$. Then, $\overline {a}_{i_{2}} = a_{i_{2}}\cup \left\{j_{1}\right\}$ and $\left\vert \overline{a}_{i_{2}}\right\vert < \left\vert a_{i_{1}}\right\vert$. Then, $\left\vert a_{i_{1}}\right\vert \geq \left\vert a_{i_{2}}\right\vert + 2$. But then, by Lemma \ref{LemCdom}(a) we obtain that $\overline{a}$ $C$-balance dominates $a$. This contradicts that $a\in \kappa \left(C,A^{f}\right)$. Therefore, $a$ is a Nash equilibrium.\qed
\end{proof}

\noindent\underline{\textit{Missing proof from Theorem \ref{TeoEx} that $a$ satisfies axiom} \textbf{A2} \textit{(Harks et al. \cite{HARKS}):}}

\begin{proof}
Suppose that $a$ is not Pareto efficient. Then the grand coalition $N$ blocks $a$. Let $\overline{a}\in A^{f}$ be the allocation induced. Let $k(a)=(x_{s})_{s=1}^{t}$ and $k(\overline{a})=(\overline{x}_{s})_{s=1}^{t}$. Note that $x_{1}\leq\overline{x}_{1}$: Otherwise, we get $k(\overline{a}) \prec k(a)$, which contradicts that $a\in \kappa \left(C,A^{f}\right)$. Also, note that we cannot have $x_{1}<\overline{x}_{1}$: Otherwise, under $\overline{a}$, the agents assigned to resources with cardinality $\overline{x}_{1}$ are worse off, contradicting that the grand coalition is better off under $\overline{a}$. Therefore, we obtain that $x_{1}=\overline{x}_{1}$. But it is easy to see that, by using similar arguments, we can iteratively show that $x_{2}=\overline{x}_{2}$, $x_{3}=\overline{x}_{3}$, and so on. Thus, we obtain that $k(a)=k(\overline{a})$. But then no agent can be better off under $\overline{a}$ unless there is an agent who is worse off under $\overline{a}$. This contradicts that the grand coalition is better off under $\overline{a}$. Therefore, $a$ is Pareto efficient. \qed
\end{proof}

\end{document}